\documentclass{article}
\usepackage{amsmath}%
\usepackage{amsthm}%
\usepackage{amsfonts}%
\usepackage{amssymb}%
\usepackage{graphicx}%
\usepackage{algorithm}
\usepackage{algorithmic}%
\usepackage{supertabular}

\def\GF{\mathrm{GF}}
\def\S{\mathrm{S}}

\def\w{\mathrm{w}}

\theoremstyle{plain}
\newtheorem{theorem}{Theorem}
\newtheorem{definition}[theorem]{Definition}
\theoremstyle{definition}
\newtheorem{proposition}[theorem]{Proposition}

\newtheorem{remark}[theorem]{Remark}
\newtheorem{corollary}[theorem]{Corollary}
\newtheorem{example}[theorem]{Example}

\begin{document}
\title{Higher Order Differentiation over Finite Fields with
Applications to Generalising the Cube Attack}

\author{Ana~S\u{a}l\u{a}gean\thanks{A. S\u{a}l\u{a}gean and R. Winter are with the Department of Computer Science,   Loughborough University,
    Loughborough, UK,
    email: \{A.M.Salagean, R.Winter\}@lboro.ac.uk}  \and Matei~Mandache-S\u{a}l\u{a}gean\thanks{M. Mandache-S\u{a}l\u{a}gean is with Trinity
College, University of Cambridge, UK,  email: {mfm41@cam.ac.uk}}
 \and Richard~Winter\footnotemark[1] \and Raphael~C.-W.~Phan
\thanks{R. Phan is with the Faculty of Engineering, Multimedia University, Malaysia,
email: {raphael@mmu.edu.my}}}

\maketitle

\begin{abstract}

Higher order differentiation was introduced in a cryptographic
context by Lai. Several attacks can be viewed in the context of
higher order differentiations, amongst them the cube attack and the
AIDA attack. All of the above have been developed for the binary
case.

We examine differentiation in larger fields, starting with the field
$\GF(p)$ of integers modulo a prime $p$. We prove a number of
results on differentiating polynomials over such fields and then
apply these techniques to generalising the cube attack to $\GF(p)$.
The crucial difference is that now the degree in each variable can
be higher than one, and our proposed attack will differentiate
several times with respect to each variable (unlike the classical
cube attack and its larger field version described by Dinur and
Shamir, both of which differentiate at most once with respect to
each variable).

Finally we describe differentiation over finite fields $\GF(p^m)$
with $p^m$ elements and prove that it can be reduced to
differentiation over $\GF(p)$, so a cube attack over $\GF(p^m)$
would be equivalent to cube attacks over $\GF(p)$.
\end{abstract}


\section{Introduction}

The original motivation for this work was to generalise the cube
attack from the binary field to arbitrary finite fields. While doing
so, we developed a number of tools and results for differentiation
over finite fields which could have a broader applicability in
cryptography.

Higher order differentiation was introduced in a cryptographic
context by Lai in~\cite{Lai94} (called there higher order
derivative). This notion had already been used for a very long time,
under the name of finite difference, in other areas of mathematics
(notably for the numerical approximation of the derivative).

The finite difference for a function $f$ is defined as the function
$(\Delta_a f)(\mathbf{x}) = f(\mathbf{x}+\mathbf{a})-f(\mathbf{x})$,
for a fixed difference $\mathbf{a}$ (the domain and codomain of $f$
are commutative groups in additive notation). Usually $f$ is a
function of $n$ variables, so $\mathbf{x} = (x_1, \ldots, x_n)$ and
$\mathbf{a} = (a_1, \ldots, a_n)$. An important particular case is
the finite difference with respect to one variable, namely
$\mathbf{a}=h\mathbf{e}_i$ where the difference step $h$ is a scalar
constant (equal to 1 by default) and $\mathbf{e}_i$ are the
elementary vectors having a 1 in position $i$ and zeroes elsewhere.
Higher order differentiation means repeated application of the
finite difference operator.

The functions we use here are functions in several variables over a
finite field. Any such function can be represented as a polynomial
function and after a sufficiently high number of applications of the
finite difference operator the result is the identically zero
function. However for certain choices of differences $\mathbf{a}$,
this can happen prematurely, for example over the binary field
$\GF(2)$ differentiating twice using the same difference
$\mathbf{a}$ will always result in the zero function, regardless of
the original function $f$. For our applications, we need to ensure
that this does not happen prematurely.

A number of cryptographic attacks can be reformulated using higher
order differentiation. Differential cryptanalysis (introduced by
Biham and Shamir~\cite{BihSha91}) has been thus reformulated by Lai
in~\cite{Lai94}; the cube attack of Dinur and Shamir~\cite{DinSha09}
and the related AIDA attack of Vielhaber~\cite{Vie07} have been
reformulated in Knellwolf and Meier~\cite{KneMei12}, Duan and
Lai~\cite{DuaLai11}.

Our main motivation came from the cube attack. In both the cube
attack and the AIDA attack we have a ``black box'' function $f$ in
several public and secret variables and we select a set of indices
of public variables $I=\{i_1, \ldots, i_k\}$. Then $f$ is evaluated
at each point of a ``cube'' consisting of the vectors that have all
the possible combinations of 0/1 values for the variables with index
in $I$, whereas the remaining variables  are left indeterminate; the
resulting values are summed and the sum will be denoted $f_I$. The
attacks hope that for suitable choices of subsets $I$ of public
variables, the resulting $f_I$ is linear in the secret variables,
for the cube attack (or equals to one secret variable or the product
of several secret variables for the AIDA attack). This situation is
particularly likely when the cardinality of $I$ is just marginally
lower than the total degree of the function. Such subsets $I$ are
found in a preprocessing phase, where the values of the keys can be
chosen freely. In the online phase the key variables are unknown,
and by computing the $f_I$ for the sets $I$ identified in the
preprocessing phase, one obtains a system of linear equations in the
key variables.

It was shown (see Knellwolf and Meier~\cite{KneMei12}, Duan and
Lai~\cite{DuaLai11}) that choosing the variable indices $I=\{i_1,
\ldots, i_k\}$ and computing $f_I$ (as described above) is
equivalent to computing the $k$-order finite difference of $f$ with
respect to the elementary vectors $\mathbf{e}_{i_1}, \ldots,
\mathbf{e}_{i_k}$, i.e. by differentiating once with respect to
$x_{i_1}$, then w.r.t. $x_{i_2}$ and so on, finally differentiating
w.r.t. $x_{i_k}$.

All the attacks above, as well as the higher order differentiation
used in cryptography are over the binary field. While all
cryptographic functions can be viewed as binary functions, there are
a number of ciphers which make significant use of operations modulo
a prime $p>2$ in their internal processing, for example
ZUC~\cite{Zuc11}, IDEA~\cite{LaiMas90,LMM91}, MMB~\cite{DGV93}. It
may therefore be advantageous for such ciphers to also be viewed and
analysed as functions over $\GF(p)$, the field of integers modulo
$p$. Unlike the binary case, a polynomial function can now have
degree more than one in each variable, in fact it can have degree up
to $p-1$ in each variable. There are yet other ciphers which use
operations over Galois fields of the form $\GF(p^m)$, for example
SNOW~\cite{SNOW} and in such fields the degree  of the polynomial
functions can be up to $p^m-1$ in each variable.

A first generalisation of the cube attack to $\GF(p^m)$ was sketched
by Dinur and Shamir in~\cite[page~284]{DinSha09} and also developed
more explicitly by Agnese and Pedicini~\cite{AgnPed11}. We show that
their approach can again be viewed as $k$-order differentiation,
where we differentiate once with respect to each of the variables
$x_{i_1}, \ldots, x_{i_k}$. However we argue that their
generalisation, while correct, has very low chances to lead to a
successful attack because we don't differentiate sufficiently many
times. Namely, on one hand, like in the binary case, the best
chances of success are when the function is differentiated a number
of times just marginally lower than its total degree; on the other
hand in their proposed scheme the number of times that the function
is differentiated is upper bounded by the number of variables, which
(unlike the binary case) can be significantly lower then the degree
of the function (see Remark~\ref{rem:cube-p-classical-not working}).

Our proposed generalisation of the cube attack to $\GF(p)$ improves
the chances of success by differentiating several times with respect
to each of the  chosen variables. Thus there is no intrinsic limit
on the number of differentiations  and therefore this number can be
as close as we want to the degree of the polynomial (only limited by
the computing power available).

We first examine higher order differentiation in $\GF(p)$
(Section~\ref{sec:diff-p}). We show that for repeated
differentiation with respect to the same variable, we can use any
non-zero difference steps and the degree will decrease by exactly
one for each differentiation. Choosing all the steps equal to one
gives a compact and efficient formula for evaluating the higher
order differentiation for a ``black box'' function.

We then show, in Section~\ref{sec:fundam-cube-p} that the main
result of the classical cube attack, \cite[Theorem 1]{DinSha09}, no
longer holds when we differentiate repeatedly with respect to the
same variable in $\GF(p)$;
Example~\ref{ex:not-direct-generalisation} gives a counterexample.
However, we show that a similar result does hold, see
Theorem~\ref{thm:several-var-fundamental-result}. Also, just like in
the binary case, if the ``black box'' function has total degree $d$,
differentiating $d-1$ times with respect to public variables always
results in a function which is either constant or is linear in the
secret variables. The resulting algorithm is sketched in
Section~\ref{sec:algorithm}. Now we not only choose variables for
the ``cube'' but we also choose the number of times we are going to
differentiate with respect to each variable. For computational
efficiency, choosing only one variable (or a small number of
variables) and differentiating  a large number of times with respect
to that variable is preferable. In $\GF(p)$ probabilistic linearity
testing has a smaller expected number of tests than in $\GF(2)$,
see~\cite{KauRon04}.

While this paper concentrates on generalising the cube attack, other
attacks that use differentiation could also be generalised to
$\GF(p)$ using our technique, for example cube testers
(see~\cite{ADMS09}) or differential cryptanalysis.

Finally, for completeness, we deal with generalisations to finite
fields of the form $\GF(p^m)$ in Section~\ref{sec:diff-pm}. Here,
for functions such as $x^d$ with $p \mid d$,  differentiation with
respect to $x$ decreases the degree by more than one regardless of
the difference step. We give a more precise expression of the
decrease in degree for higher order differentiation depending on the
representation of the degree in base $p$. Any function can be
differentiated at most $m(p-1)$ times before it becomes identically
zero. Moreover, in order to avoid the result becoming identically
zero even earlier, the difference steps will be chosen as follows:
$p-1$ steps equal to $b_0$, $p-1$ steps equal to $b_1$ and so on,
where $b_0,\ldots, b_{m-1}$ is a base of $\GF(p^m)$ when viewed as a
vector space over $\GF(p)$. We can thus differentiate $m(p-1)$
times. Due to the fact that differentiation only uses the additive
group of $\GF(p^m)$, which is isomorphic to $\GF(p)^m$,
differentiation over $\GF(p^m)$ can in fact be reduced to
differentiation over $\GF(p)$ in each component of the projection of
the function $f$. Therefore, we feel that developing a cube attack
in $\GF(p^m)$, while possible, does not bring any additional
advantages compared to a cube attack in $\GF(p)$.

\section{Preliminaries}
Throughout this paper $R$ denotes an arbitrary commutative ring with
identity and $\GF(p^m)$ denotes the finite field with $p^m$ elements
where $p$ is prime. We denote by $\mathbf{e}_i = (0,\ldots, 0, 1, 0,
\ldots, 0) \in R^n$ the vector which has a 1 in position $i$ and
zeroes elsewhere, i.e. $\mathbf{e}_1, \ldots, \mathbf{e}_n$ is the
canonical basis of the vector space $R^n$.

We recall the definition of differentiation, which was introduced in
the cryptographic context by Lai in~\cite{Lai94}. This notion was
used long before, under the name finite difference, in other areas
of mathematics, notably for approximating the derivative.
\begin{definition} \label{def:diff}
Let $f: R^n \rightarrow R^s$ be a function in $n$ variables
$x_1,\ldots, x_n$. Let $\mathbf{a} = (a_1,\ldots, a_n) \in
R^n\setminus \{\mathbf{0}\}$. The {\em finite difference operator}
(or differentiation operator) with respect to $\mathbf{a}$
associates to each function $f$ the function $\Delta_{\mathbf{a}} f
: R^n \rightarrow R^s$ defined as
 \[
 \Delta_{\mathbf{a}} f(x_1,\ldots,x_n) = f(x_1 + a_1, \ldots, x_n + a_n)-
 f(x_1,\ldots,x_n).
 \]
Denoting $\mathbf{x} = (x_1,\ldots, x_n)$ we can also write
$\Delta_{\mathbf{a}} f(\mathbf{x}) = f(\mathbf{x+a}) -
f(\mathbf{x})$.

For the particular case of $a=h\mathbf{e_i}$ for some $1\le i\le n$
and $h\in R\setminus\{0\}$, we will call
$\Delta_{\mathbf{h\mathbf{e_i}}}$ the finite difference operator (or
differentiation) with respect to the variable $x_i$ with step $h$,
or simply the finite difference operator with respect to the
variable $x_i$ if $h=1$ or if $h$ is clear from the context. We will
use the abbreviation ``w.r.t. $x_i$'' for ``with respect to $x_i$''.
\end{definition}

\begin{remark}
Note that in the cryptographic literature this operator (and the
resulting function) is usually called the derivative or differential
(see~\cite{Lai94,Knu95}).
 We will avoid the term derivative because of the risk of
confusion with the well established mathematical notion of {\em
formal derivative} of a polynomial. For a polynomial $ \sum_{i=0}^d
c_i x^i\in R[x]$ the formal derivative w.r.t. $x$ is defined as
$\sum_{i=1}^d c_i i x^{i-1}$. It can easily be seen that the formal
derivative operator w.r.t. $x_i$ coincides with the finite
difference operator w.r.t. $x_i$ only for polynomials which have
degree at most one in $x_i$. Polynomial functions over $\GF(2)$ have
degree at most one in each variable, so in this case these notions
coincide. Hence the use of the term ``derivative'' for
$\Delta_{h\mathbf{e_i}} f(x_1,\ldots,x_n)$ is justified for
polynomials over $\GF(2)$, but not for polynomials over other
rings/fields.
\end{remark}

\begin{remark}
For defining the finite difference operator, we do not actually need
to work over a ring $R$, a commutative group (using additive
notation for convenience) is sufficient. Here we used a ring due to
our application to finite fields, and also due to some of the
techniques involving polynomials.
\end{remark}

The finite difference operator is a linear operator; it is
commutative and associative. Repeated application of the operator
(also called higher order differentiation or higher order derivative
in~\cite{Lai94}) will be denoted by
\[
\Delta^{(k)}_{\mathbf{a_1}, \ldots, \mathbf{a_k}} f =
\Delta_{\mathbf{a_1}}\Delta_{\mathbf{a_2}}\ldots \Delta_\mathbf{a_k}
f
\]
where $\mathbf{a_1}, \ldots, \mathbf{a_k} \in R^n$ are not
necessarily distinct.
An explicit formula can be obtained easily from
Definition~\ref{def:diff} by induction:
\begin{proposition} \label{prop:diff-incl-excl-formula}
Let $f: R^n \rightarrow R^s$ be a function in $n$ variables
$x_1,\ldots,x_n$. Let  $\mathbf{a_1}, \ldots, \mathbf{a_k} \in
R^n\setminus \{ \mathbf{0} \}$ not necessarily distinct. Then
\[
\Delta^{(k)}_{\mathbf{a_1}, \ldots, \mathbf{a_k}} f(\mathbf{x}) =
\sum_{j=0}^k (-1)^{k-j} \sum_{\{i_1, \ldots, i_j\} \subseteq \{1,
\ldots, k\}} f(\mathbf{x+a_{i_1} + \cdots + a_{i_j}}).
\]
\end{proposition}

Depending of the values of the $\mathbf{a_1}, \ldots, \mathbf{a_k}$
and the characteristic of the ring, $\Delta^{(k)}_{\mathbf{a_1},
\ldots, \mathbf{a_k}} f$ could collapse, becoming the identical zero
function regardless of the function $f$. (This happens, for example,
if the ring is $\GF(2)$ and $\mathbf{a_1}, \ldots, \mathbf{a_k}$ are
not linearly independent.) When differentiating w.r.t. one variable
we need to choose the difference steps so that that this does not
happen. Details will be given in Section~\ref{sec:diff-p} for finite
fields of the form $\GF(p)$ and in Section~\ref{sec:diff-pm} for
finite fields of the form $\GF(p^m)$.

While the finite difference operator can be defined for any
function, in the sequel we will concentrate on polynomial functions.
We will denote by $\deg_{x_i}(f)$ the degree of $f$ in the variable
$x_i$. The total degree will be denoted $\deg(f)$. The following
three results are well known and straightforward, but will be needed
later. The first result states that differentiating with respect to
one variable decreases the degree in that variable by at least one.
The other propositions deal with the results of the differentiation
in a few simple cases.

\begin{proposition}\label{prop:deg-decreases-by-at least-one}
Let $f: R^n \rightarrow R$ be a polynomial function. Let
$h,h_1,\ldots, h_k\in R\setminus\{0\}$ and $i\in\{1, \ldots, n\}$.
If $\deg_{x_i}(f)=0$ then $\Delta_{h\mathbf{e_i}} f(x_1,\ldots,x_n)
\equiv 0$. If $\deg_{x_i}(f)>0$ then
$\deg_{x_i}(\Delta_{h\mathbf{e_i}} f) \le \deg_{x_i}(f)-1$.
Consequently $\deg_{x_i}(\Delta^{(k)}_{h_1\mathbf{e_i}, \ldots,
h_k\mathbf{e_i}} f) \le \deg_{x_i}(f)-k$ if $k\le\deg_{x_i}(f)$, and
$\Delta^{(k)}_{h_1\mathbf{e_i}, \ldots, h_k\mathbf{e_i}} f$ is
identically zero if $k>\deg_{x_i}(f)$.
\end{proposition}

\begin{proposition}\label{prop:diff-degree1}
Let $h\in R$, $h \ne 0$ and $i\in\{1, \ldots, n\}$, Let $f: R^n
\rightarrow R$ be a polynomial function with $\deg_{x_i}(f)=1$ i.e.
$f(x_1,\ldots,x_n) = x_i g_1(x_1,\ldots, x_{i-1}, x_{i+1}, \ldots,
x_n)+ g_2(x_1,\ldots, x_{i-1}, x_{i+1}, \ldots, x_n)$ ($g_1$ and
$g_2$ are polynomial functions that do not depend on $x_i$). Then
\[\Delta_{h\mathbf{e_i}} f(x_1,\ldots,x_n) = h g_1(x_1, \ldots,
x_{i-1}, x_{i+1}, \ldots, x_n).\]
\end{proposition}

\begin{proposition}\label{prop:one-diff-at-zero}
Let $f: R^n \rightarrow R$ be a polynomial function. Let $h\in R$,
$h \ne 0$ and $i\in\{1, \ldots, n\}$. Factoring out $x_i$  we write
$f(x_1,\ldots,x_n) = x_i g_1(x_1,\ldots, x_n)+ g_2(x_1,\ldots,
x_{i-1}, x_{i+1}, \ldots, x_n)$ ($g_2$ is a polynomial function that
does not depend on $x_i$, but $g_1$ may depend on $x_i$). Then
\[(\Delta_{h\mathbf{e_i}} f)(x_1,\ldots, x_{i-1}, 0, x_{i+1}, \ldots,
x_n) = h g_1(x_1,\ldots, x_{i-1}, h, x_{i+1}, \ldots, x_n).\]
\end{proposition}

Recall that for integers $d, k_1, k_2, \ldots, k_s$ such that
$\sum_{i=1}^s k_i = d$ and $k_i\ge 0$ the multinomial is defined as:
\[
\binom{d}{k_1, k_2, \ldots, k_s} = \frac{d!}{k_1! k_2! \cdots k_s!}.
\]
One combinatorial interpretation is the number of ways that we can
distribute $n$ objects into $s$ (labeled) boxes, so that the first
box has $k_1$ elements, the second $k_2$ elements e.t.c.
 Multinomials are
generalisations of the usual binomial coefficients, with
\[
\binom{d}{k} = \binom{d}{k, d-k}.
\]
Next we examine the effect of higher order differentiation on
univariate monomials; the general formula for univariate polynomials
can be obtained using the linearity of the $\Delta$ operator.
\begin{theorem}\label{thm:diff-formula}
Let $f: R \rightarrow R$ defined by $f(x) = x^d$. Let $
h_1,\ldots,h_k \in R\setminus\{0\}$
\[
\Delta^{(k)}_{h_1\mathbf{e_1}, \ldots, h_k\mathbf{e_1}} x^d =
\sum_{j=k}^d \left(\sum_{\substack{(i_1, \ldots, i_k) \in \{1, 2,
\ldots, j-k+1\}^k \\i_1+\ldots +i_k=j}} \binom{d}{i_1, \ldots, i_k,
d-j} h_1^{i_1}\cdots h_k^{i_k}\right) x^{d-j}
\]
\end{theorem}
\begin{proof}
Induction on $k$. For $k=1$ we have $\Delta_{h_1\mathbf{e_1}} x^d =
(x+h_1)^d - x^d = \sum_{j=1}^d \binom{d}{j} h_1^j x^{d-j} =
\sum_{j=1}^d \binom{d}{j, d-j} h_1^j x^{d-j}$ and the statement is
verified.

Not let us assume the statement holds for a given $k$ and we prove
it for $k+1$.
\begin{eqnarray*}
\lefteqn{\Delta^{(k+1)}_{h_1\mathbf{e_1}, \ldots,
h_{k+1}\mathbf{e_1}} x^d  =  \Delta_{h_{k+1}
\mathbf{e_1}}\Delta^{(k)}_{h_1\mathbf{e_1},
\ldots, h_{k}\mathbf{e_1}} x^d }\\
 & = & \sum_{j=k}^d
\left(\sum_{\substack{(i_1, \ldots, i_k) \in \{1, 2, \ldots,
j-k+1\}^k \\i_1+\ldots +i_k=j}} \binom{d}{i_1, \ldots, i_k, d-j}
h_1^{i_1}\cdots h_k^{i_k}\right) ((x+h_{k+1})^{d-j} -x^{d-j}) \\
& = & \sum_{j=k}^d \left(\sum_{\substack{(i_1, \ldots, i_k) \in \{1,
2, \ldots, j-k+1\}^k \\i_1+\ldots +i_k=j}} \binom{d}{i_1, \ldots,
i_k, d-j} h_1^{i_1}\cdots
h_k^{i_k}\right)\left(\sum_{i_{k+1}=1}^{d-j}
\binom{d-j}{i_{k+1}} h_{k+1}^{i_{k+1}} x^{d-j-i_{k+1}}\right) \\
& = & \sum_{j'=k+1}^d \left(\sum_{\substack{(i_1, \ldots, i_{k+1})
\in \{1, 2, \ldots, j'-k\}^{k+1} \\i_1+\ldots +i_{k+1}=j'}}
\binom{d}{i_1, \ldots, i_{k+1}, d-j'} h_1^{i_1}\cdots
h_{k+1}^{i_{k+1}}\right) x^{d-j'}.
\end{eqnarray*}
The last line uses the identity:
\[ \binom{d}{i_1, \ldots, i_k, d-j}\binom{d-j}{i_{k+1}} =\binom{d}{i_1, \ldots, i_{k+1},
d-j'}\] where $j=i_1+\ldots +i_k$ and $j'= j+i_{k+1}$. 
\end{proof}

Recall that in a finite field $\GF(p^m)$ we have $a^{p^m} = a$ for
all elements $a$. Hence, while (formal) polynomials over $\GF(p^m)$
could have any degree in each variable, when we talk about the
associated { \em polynomial function}, there will always be a unique
polynomial of degree at most $p^m-1$ in each variable which defines
the same polynomial function. In other words we are working in the
quotient ring $\GF(p^m)[x_1, \ldots,x_n]/\langle x_1^{p^m} - x_1,
\ldots x_n^{p^m} - x_n \rangle$, and we use as representative of
each class the unique polynomial which has degree at most $p^m-1$ in
each variable.

Moreover, all functions in $n$ variables over a finite field can be
written as polynomial functions of $n$ variables  The polynomial can
be obtained by interpolation from the values of the function at each
point in its (finite) domain. (This is obviously not the case for
infinite fields). To summarise, each function in $n$ variables over
$\GF(p^m)$ can be uniquely expressed as a polynomial function
defined by a polynomial in $\GF(p^m)[x_1, \ldots, x_n]$ of degree at
most $p^m-1$ in each variable.

\section{Classical cube attack and differentiation}\label{sec:cube-classical}
In this section we first recall the classical cube attack
from~\cite{DinSha09}, and its interpretation in the framework of
higher order differentials (see~\cite{KneMei12,DuaLai11}). We then
recall a first generalisation to higher fields sketched
in~\cite{DinSha09} (see also~\cite{AgnPed11}).

In the cube attack (\cite{DinSha09}), one has a ``black box''
polynomial function  $f: \GF(2)^n \rightarrow \GF(2)$ in $n$
variables $x_1,\ldots, x_n$. Recall that polynomial functions over
$\GF(2)$ have degree at most one in each variable. (Note that the
function is named $p$ in the cube attack papers, but we had to
rename it $f$ as later we will work in fields of characteristic
other than 2, and we felt $p$ was a well-established notation for
the characteristic.)

The next definitions are taken from \cite{DinSha09}: ``Any subset
$I$ of size $k$ defines  a $k$ dimensional Boolean cube $C_I$ of
$2^k$ vectors in which we assign all the possible combinations of
0/1 values to variables in $I$ and leave all the other variables
undetermined. Any vector $\mathbf{v} \in C_I$ defines a new derived
polynomial $f_{|\mathbf{v}}$ with $n-k$ variables (whose degree may
be the same or lower than the degree of the original polynomial).
Summing these derived polynomials over all the $2^k$ possible
vectors in $C_I$ we end up with a new polynomial which is denoted by
$ f_I = \sum_{\mathbf{v} \in C_I} f_{|\mathbf{v}}. $.'' Note that
the computation of $f_I$ requires $2^k$ calls to the ``black box''
function $f$. On the other hand denoting by $t_I$ the product of the
variables with indices in $I = \{i_1, \ldots, i_k\}$, i.e.
$t_I=x_{i_1}\cdots x_{i_k}$, we can factor the common subterm $t_I$
out of some of the terms in $f$ and write $f$ as
\[
f(x_1, \ldots, x_n) = t_I f_{S(I)} + r(x_1, \ldots, x_n).
\]
where each of the terms of $r(x_1, \ldots, x_n)$ misses at least one
of the variables with index in $I$. Note that $f_{S(I)}$ is a
polynomial in the variables with indices in $\{1,2, \ldots, n\}
\setminus I$.

The cube attack is based on the following main result:

\begin{theorem}(\cite[Theorem
1]{DinSha09})\label{thm:cube-classical} For any polynomial $f$ and
subset of variables $I$, $f_I \equiv f_{S(I)} \pmod 2$.
\end{theorem}

This result was reformulated using higher order differentials by
several authors (\cite{KneMei12,DuaLai11}). We present such a
reformulation using our notations:

\begin{theorem}\label{thm:cube-classical-reformulated} For any polynomial $f$ and
subset of variables $I = \{ i_1, \ldots, i_k\}$,  we have $
\Delta^{(k)}_{\mathbf{e}_{i_1}, \ldots, \mathbf{e}_{i_k}} f =f_I =
f_{S(I)}$.
\end{theorem}
\begin{proof}
To show that  $f_I = \Delta^{(k)}_{\mathbf{e}_{i_1}, \ldots,
\mathbf{e}_{i_k}} f$ we use
Proposition~\ref{prop:diff-incl-excl-formula}. We have
\[
\Delta^{(k)}_{\mathbf{e_{i_1}}, \ldots, \mathbf{e_{i_k}}}
f(\mathbf{x}) = \sum_{(b_1,\ldots,b_k)\in GF(2)^k} f(x_1, \ldots,
x_{i_1-1}, x_{i_1}+b_1, x_{i_1+1}, \ldots, x_{i_k}+b_k, \ldots, x_n)
\]
Note that evaluating the expression above for any fixed constant
values of $x_{i_1}, \ldots, x_{i_k}$ yields $f_I$. Hence
$\Delta^{(k)}_{\mathbf{e}_{i_1}, \ldots, \mathbf{e}_{i_k}} f$  does
not depend on $x_{i_1}, \ldots, x_{i_k}$ and is equal to $f_I$. By
Theorem~\ref{thm:cube-classical}, $f_I = f_{S(I)}$.
\end{proof}

For the cube attack we are particularly interested in the situation
when $f_{S(I)}$ (and therefore $f_I$) has degree exactly one, i.e.
it is linear but not constant (the corresponding term $t_I$ is then
called maxterm in \cite{DinSha09}). Let $d$ be the total degree of
$f$. Then $I$ having $d-1$ elements is a sufficient (but not
necessary) condition for $f_{S(I)}$ to have degree at most one, i.e.
to be linear or constant.

Generalising the cube attack from the binary field to $\GF(p^m)$ was
sketched in \cite{DinSha09}: ``Over a general field $\GF(p^m)$ with
$p>2$, the correct way to apply cube attacks is to alternately add
and subtract the outputs of the master polynomial with public inputs
that range only over the two values 0 and 1 (and not over all their
possible values of $0, 1, \ldots, p$), where the sign is determined
by the sum (modulo 2) of the vector of assigned values.''

We make this idea more precise; this was also done in
\cite{AgnPed11} but we will follow a simpler approach for the proof
of the main result. Let $f$ be again a function of $n$ variables
$x_1,\ldots, x_n$, but this time over an arbitrary finite field
$\GF(p^m)$. Note that now $f$ can have degree up to $p^m-1$ in each
variable.

As before, we select a subset of $k$ indices $I = \{i_1, \ldots,
i_k\}\subseteq \{1,2, \ldots, n\}$ and consider a ``cube'' $C_I$
consisting of the $n$-tuples which have all combinations of the
values 0/1 for the variables with indices in $I$, while the other
variables remain indeterminate. The function $f$ is evaluated at the
points in the cube and these values are summed with alternating +
and $-$ signs obtaining a value
\[
f_I = \sum_{\mathbf{v} \in C_I} (-1)^{k-\w(\mathbf{v})}
f_{|\mathbf{v}}
\]
where $\w(\mathbf{v})$ denotes the Hamming weight of $\mathbf{v}$
ignoring the variables what have remained indeterminate.

On the other hand denoting by $t_I$ the product of the variables
with indices in $I$, we can factor the common subterm $t_I$ out of
some of the terms in $f$ and write $f$ as
\[
f(x_1, \ldots, x_n) = t_I f_{S(I)}(x_1, \ldots, x_n) + r(x_1,
\ldots, x_n).
\]
where each of the terms of $r(x_1, \ldots, x_n)$ misses at least one
of the variables with index in $I$. Note that, unlike the binary
case, now $f_{S(I)}$ can contain variables with indices in $I$.

Now we can prove an analogue of Theorems~\ref{thm:cube-classical}
and~\ref{thm:cube-classical-reformulated}. (A similar theorem
appears in~\cite[Theorem 6]{AgnPed11}, but both the statement and
the proof are more complicated, involving a term $t =
x_{i_1}^{r_1}\cdots x_{i_k}^{r_k}$ instead of $x_{i_1}\cdots
x_{i_k}$ and consequently when factoring out $t$ and writing $f = t
f_{S(t)} + q$, having to treat separately the terms of $q$ which
contain some variables with indices in $I$, and the terms of $q$
which do not.)
\begin{theorem}\label{thm:fundamental-p-once}
Let $f: \GF(p^m)^n \rightarrow \GF(p^m)$ be a polynomial function
and $I$ a subset of variable indices. Denote by $ \mathbf{v}$ the
$n$-tuple having values of 1 in the positions with indices in $I$
and indeterminates in the other positions, and by  $ \mathbf{u}$ the
$n$-tuple having values of 0 in the positions in $I$ and
indeterminates in the other positions. Then:
\[
f_I = (\Delta^{(k)}_{\mathbf{e}_{i_1}, \ldots, \mathbf{e}_{i_k}}
f)(\mathbf{u})= f_{S(I)}(\mathbf{v})
\]
\end{theorem}
\begin{proof}
The fact that $f_I = (\Delta^{(k)}_{\mathbf{e}_{i_1}, \ldots,
\mathbf{e}_{i_k}} f)(\mathbf{u})$ follows from
Proposition~\ref{prop:diff-incl-excl-formula} in the same way as in
the proof of Theorem~\ref{thm:cube-classical-reformulated}.

It suffices to show
 $(\Delta^{(k)}_{\mathbf{e}_{i_1},
\ldots, \mathbf{e}_{i_k}} f)(\mathbf{u})= f_{S(I)}(\mathbf{v})$ for
the case when $f$ is a monomial. The rest follows from the linearity
of the operators, as $(f+g)_{S(I)} = f_{S(I)}+g_{S(I)}$ and $\Delta$
is a linear operator.

If $f$ is a monomial not divisible by $t_I$, then both $f_{S(I)}$
and $\Delta^{(k)}_{\mathbf{e}_{i_1}, \ldots, \mathbf{e}_{i_k}} f$
are identically zero, the latter using
Proposition~\ref{prop:deg-decreases-by-at least-one}.

Now assume $f = t_I f_{S(I)}$ for some monomial $f_{S(I)}$. Like in
the proof of~\cite[Theorem 1]{DinSha09}, we note that in the sum in
the definition of $f_I$ (or, equivalently in the sum given by
Proposition~\ref{prop:diff-incl-excl-formula} for
$(\Delta^{(k)}_{\mathbf{e}_{i_1}, \ldots, \mathbf{e}_{i_k}}
f)(\mathbf{u})$) only one term is non-zero; namely  $t_I$ evaluates
to a non-zero value iff $x_{i_1}= \ldots = x_{i_k} = 1$; hence $f_I
= f_{S(I)}(\mathbf{v})$.
 \end{proof}

\begin{remark}\label{rem:cube-p-classical-not working}
A cube attack based on Theorem~\ref{thm:fundamental-p-once} above
would again search for sets $I$ for which $f_I$ is linear in the
variables whose indices are not in $I$. If the total degree of $f$
is $d$, the total degree of $f_I$  can be, in the worst case, $d-k$.
If $k=d-1$ we can guarantee that $f_I$ is linear or a constant. More
generally, the closer $k$ gets to $d-1$ (while still having $k\le
d-1$), the higher the chances of linearity.

However, unlike the binary case where $d\le n$, now $d$ can have any
value up to $(p^m-1)n$. Hence the (unknown) degree $d$ of $f$ could
well be considerably higher than the number of variables $n$. In
such a case, $k\le n-1$ is considerably lower than $d-1$, and the
chances of linearity are very small. In other words, since we
differentiate at most once w.r.t. each variable, so a total of at
most $n-1$ times, the degree decreases by around $n-1$ in general,
and the resulting function can still have quite a high degree.
Therefore, while a cube attack based on this result would be
correct, it would have extremely low chances of success.

Our proposed generalisation of this attack would increase these
chances by differentiating several times with respect to each
variable. This will result in a greater decrease of the degree, thus
improving the chances of reaching a linear result.
\end{remark}

\section{Generalisations  to $\GF(p)$} \label{sec:proposed-cube}
\subsection{Differentiation in $\GF(p)$}\label{sec:diff-p}

Differentiation with respect to a variable decreases the degree in
that variable by at least one. In the binary case, the degree of a
polynomial function in each variable is at most one, so we can only
differentiate once w.r.t. each variable; a second differentiation
will trivially produce the zero function. In $\GF(p)$ the degree in
each variable is up to $p-1$. We can therefore consider
differentiating several times (and possibly using different
difference steps) with respect to each variable. We first show that
a monomial of degree $d_i$ in a variable $x_i$ can be differentiated
$m_i$ times w.r.t. $x_i$, for any $m_i\le d_i$ (and using any
collection of non-zero difference steps) and the degree decreases by
exactly $m_i$. Hence we can differentiate $d_i$ times without the
result becoming identically zero.

\begin{theorem}\label{thm:one-var-poly}
Let $m_1\le d_1\le p-1$ and $h_1,\ldots, h_{m_1} \in \GF(p)\setminus
\{0\}$. Then
\begin{eqnarray}\label{eq:one-var-poly}
\lefteqn{
\Delta^{(m_1)}_{h_1\mathbf{e_{1}},
\ldots, h_{m_1}\mathbf{e_{1}}} x_1^{d_1} = } \nonumber \\
& = & \sum_{j=m_1}^{d_1} \left(\sum_{\substack{(i_1, \ldots,
i_{m_1}) \in \{1, 2, \ldots, j-m_1+1\}^{m_1} \nonumber
\\i_1+\ldots +i_{m_1}=j}} \binom{d_1}{i_1, \ldots, i_{m_1}, d_1-j}
h_1^{i_1}\cdots h_{m_1}^{i_{m_1}} \right) x_1^{d_1-j}
\end{eqnarray}
 In the expression above the
coefficient of $x_1^{d_1-m_1}$ equals
\[\binom{d_1}{1, 1, \ldots, 1, d_1-m_1}   h_1\cdots h_{m_1} = \frac{d_1!}{(d_1-m_1)!}
h_1\cdots h_{m_1} \neq 0
\]
hence the degree in $x_1$ is exactly $d_1-m_1$. For the particular
case of $h_1=\ldots = h_{m_1}=1$, the leading coefficient becomes
\[\binom{d_1}{1, 1, \ldots, 1, d_1-m_1}  = \frac{d_1!}{(d_1-m_1)!}
\]
and the free term becomes
\begin{equation}\label{eq:free-term-one-var}
\sum_{\substack{(i_1, \ldots, i_{m_1}) \in \{1, 2, \ldots,
d_1-m_1+1\}^{m_1}
\\i_1+\ldots +i_{m_1}=d_1}} \binom{d_1}{i_1, \ldots, i_{m_1}, 0}.
\end{equation}
If $h_1=\ldots = h_{m_1}=1$ and moreover $d_1=m_1$ we have
\begin{equation}
\label{eq:full-diff-one-var} \Delta^{(d_1)}_{\mathbf{e_{1}}, \ldots,
\mathbf{e_{1}}} x_1^{d_1}  = \binom{d_1}{1, 1, \ldots, 1, 0} \bmod p
= d_1! \bmod p.
\end{equation}
\end{theorem}
\begin{proof}
Use Theorem~\ref{thm:diff-formula}. Since $d_1<p$, we have that
$\displaystyle \frac{d_1!}{(d_1-m_1)!}$ is not divisible by $p$.
 \end{proof}

\begin{example}\label{ex:diff}
Let $f(x_1, x_2, x_3, x_4) = x_1^5 x_2 + x_1^4 x_3 x_4 + x_4^6$ be a
polynomial with coefficients in  $\GF(31)$. We choose the variable
$x_1$ and differentiate repeatedly w.r.t. $x_1$, always with
difference step equal to one. Differentiating once w.r.t. $x_1$ we
obtain:
\begin{eqnarray*}
\lefteqn{
\Delta_{\mathbf{e_{1}}} f(x_1, x_2, x_3, x_4) =} \\ & & x_1^4(5x_2)
+ x_1^3(10 x_2+4x_3x_4) + x_1^2(10 x_2+6x_3x_4) +x_1(5 x_2+4x_3x_4)
+x_2+x_3x_4
\end{eqnarray*}
Differentiating again w.r.t. $x_1$ we obtain:
\begin{eqnarray*}
\lefteqn{
\Delta^{(2)}_{\mathbf{e_{1}},\mathbf{e_{1}}} f(x_1, x_2, x_3, x_4) =
}
\\ &  & x_1^3(20x_2) + x_1^2(29 x_2+12x_3x_4)  +x_1(8 x_2+ 24
x_3x_4) + 30 x_2 + 14 x_3x_4
 \end{eqnarray*}
Finally, if we differentiate a total of 5 times we obtain:
\[
\Delta^{(5)}_{\mathbf{e_{1}},
\ldots, \mathbf{e_{1}}} f(x_1, x_2, x_3, x_4) = 5! x_2 = 27 x_2
\]
as expected by (\ref{eq:full-diff-one-var}).
\end{example}

For the remainder of this section, for simplicity we will always
choose all the difference steps $h_i$ equal to one. The case of
arbitrary $h_i$ can be treated similarly, but the formulae become
more cumbersome. For convenience we will introduce some more
notation. We pick a subset of $k$ variable indices $I= \{ i_1,
\ldots, i_k \}$ and we also pick multiplicities for each variable,
$m_1, \ldots, m_k$. Denote by $t$ the term $t = x_{i_1}^{m_1} \cdots
x_{i_k}^{m_k}$. We will apply the finite difference operator $m_1$
times w.r.t. the variable $x_{i_1}$, and $m_2$ times w.r.t. the
variable $x_{i_2}$ etc. always with difference step equal to one.
 More precisely we
define:
\[
f_t(x_1, \ldots, x_n) = \Delta^{(m_1)}_{\mathbf{e_{i_1}}, \ldots,
\mathbf{e_{i_1}}} \ldots \Delta^{(m_k)}_{\mathbf{e_{i_k}}, \ldots,
\mathbf{e_{i_k}}} f(x_1, \ldots, x_n)
\]

We now generalise Theorem~\ref{thm:one-var-poly} to the case when we
differentiate w.r.t. several variables.

\begin{theorem}\label{thm:several-var-fundamental-result-one-monomial}
 Let $k\le n$ and let
$m_1, \ldots, m_k$ and $d_1, \ldots, d_k$ be integers such that
$1\le m_{\ell} \le d_{\ell}\le p-1$ for $\ell = 1, \ldots, k$. Let
$\{i_1, \ldots,i_k\}\subseteq \{1, \ldots, n\}$ and let $f:\GF(p)^n
\rightarrow \GF(p)$, $f(x_1, \ldots, x_n) = x_{i_1}^{d_1} \cdots
x_{i_k}^{d_k}$ and $t = x_{i_1}^{m_1} \cdots x_{i_k}^{m_k}$. We have
\[ f_t  = \sum_{j_1 =
m_1}^{d_1} \ldots \sum_{j_k = m_k}^{d_k} D(d_1,j_1,m_1) \ldots
D(d_k,j_k,m_k) x_{i_1}^{d_1-j_1}\ldots x_{i_k}^{d_k-j_k}
\]
where for any $1\le m \le j \le d$ we define \[ D(d,j,m) =
\sum_{\substack{(i_1, \ldots, i_{m}) \in \{1, 2, \ldots, j-m+1\}^{m}
\nonumber
\\i_1+\ldots + i_{m}=j}} \binom{d}{i_1, \ldots, i_{m}, d-j}.
\]
The coefficient of $x_{i_1}^{d_1-m_1}\ldots x_{i_k}^{d_k-m_k}$ in
$f_t$
is equal to $\prod_{\ell=1}^k
\displaystyle{\frac{d_{\ell}!}{(d_{\ell}-m_{\ell})!}}\ne 0$ and the
free term is equal to $\prod_{\ell=1}^k D(d_{\ell}, d_{\ell},
m_{\ell})$. The total degree of $f_t$
is $\sum_{\ell=1}^k (d_{\ell} - m_{\ell})$.

For the particular case of $m_1=d_1, \ldots, m_k=d_k$, we have
\[
%
f_t  =d_1! \ldots d_k!.
\]
\end{theorem}
\begin{proof}
Induction on $k$, applying Theorem~\ref{thm:one-var-poly}. 
\end{proof}

When all the difference steps $h_i$ are equal to one, the evaluation
of the finite difference for a ``black box'' function $f$ using
Proposition~\ref{prop:diff-incl-excl-formula} becomes simpler. We
treat first the case when we differentiate w.r.t. a single variable:
\begin{proposition}\label{prop:incl-excl-one-var-p}
Let $f: R^n\rightarrow R$. Then
\begin{eqnarray}
\lefteqn{ f_{x_1^{m_1}} (x_1,
\ldots, x_n)=}\nonumber \\
\label{eq:incl-excl-one-var-p} && \Delta^{(m_1)}_{\mathbf{e_{1}},
\ldots, \mathbf{e_{1}}} f (x_1, \ldots, x_n) =  \sum_{i=0}^{m_1}
(-1)^{m_1-i}\binom{m_1}{i}f(x_1+i, x_2, \ldots, x_n).
\end{eqnarray}
If $R$ has characteristic $p$ and $m_1<p$ then all the coefficients
$\displaystyle \binom{m_1}{i}$ in the sum above are non-zero. If $f$
is a ``black box'' function, evaluating $f_{x_1^{m_1}}$ at one point
in its domain requires $m_1+1$ evaluations of $f$.
\end{proposition}
\begin{proof}
The formula follows from
Proposition~\ref{prop:diff-incl-excl-formula}. Since $0\le i \le
m_1<p$ and $p$ is prime, $\binom{m_1}{i}$ cannot be divisible by
$p$, so it is non-zero in a field of characteristic $p$. 
\end{proof}

We now look at the situation where we differentiate w.r.t. several
variables $x_{i_1}, \ldots, x_{i_k}$.
\begin{proposition} \label{prop:black-box-eval-p} Let $f:R^n \rightarrow R$ and
 $t = \prod_{\ell=1}^k x_{i_{\ell}}^{m_{\ell}}$.  Then
\[f_t(\mathbf{x}) = \sum_{j_1=0}^{m_1} \ldots \sum_{j_k=0}^{m_k}
(-1)^{\sum_{\ell=1}^k (m_{\ell}-j_{\ell})} \binom{m_1}{j_1} \cdots
\binom{m_k}{j_k} f(\ldots, x_{i_1}+j_1, \ldots, x_{i_k}+j_k,
\ldots).
\]
If $R$ has characteristic $p$ and all $m_{\ell}<p$, then all the
coefficients $\binom{m_1}{j_1} \cdots \binom{m_k}{j_k}$ in the sum
above are non-zero. If $f$ is a ``black box'' function, one
evaluation of $f_t$ needs $\prod_{\ell=1}^k (m_{\ell}+1)$
evaluations of $f$. In particular evaluating $f_t$ for $x_{i_{\ell}}
=0$, $\ell=1, \ldots, k$ we obtain:
\begin{equation}\label{eq:evaluation-diff-p}
\sum_{j_1=0}^{m_1} \ldots \sum_{j_k=0}^{m_k} (-1)^{\sum_{\ell=1}^k
(m_{\ell}-j_{\ell})} \binom{m_1}{j_1} \cdots \binom{m_k}{j_k}
f(\ldots, j_1, \ldots, j_k, \ldots)
\end{equation}
with $j_1, \ldots, j_k$ in positions $i_1, \ldots, i_k$
respectively.
\end{proposition}
We note that in terms of the time complexity of evaluating $f_t$, it
is now not only the total degree of $t$ that matters (as in the
binary case), but also the exponents of each variable. For a given
number $m$ of differentiations (i.e. $t$ of total degree $m$), the
smallest time complexity is achieved when $t$ contains only one
variable, i.e. $t=x_{i_1}^{m}$. Among all $t$ of total degree $m$
that contain $k$ variables, the best time complexity is achieved
when $t$ has degree one in each but one of its variables, e.g.
$t=x_{i_1}^{m-k+1}x_{i_2}\ldots x_{i_k}$.

\begin{remark} We saw that in the binary case, differentiating once w.r.t.
each of the variables $x_{i_1}, \ldots x_{i_k}$ is equivalent to
summing $f$ evaluated over a ``cube'' consisting of all the 0/1
combinations for the variables $x_{i_1}, \ldots x_{i_k}$. According
to Proposition~\ref{prop:black-box-eval-p}, the analogue of the
``cube'' will now be a $k$-dimensional grid/mesh with sides of
``length'' $m_1, m_2, \ldots, m_k$. Namely each variable
$x_{i_{\ell}}$ w.r.t. which we differentiate will have increments of
$0, 1, \ldots, m_{\ell}$ and each term in the sum has alternating
signs as well as being multiplied by binomial coefficients.
\end{remark}


\subsection{Fundamental theorem of the cube attack generalised to $\GF(p)$}
\label{sec:fundam-cube-p}

We will use the notation $f_t = \Delta^{(m_1)}_{\mathbf{e_{i_1}},
\ldots, \mathbf{e_{i_1}}} \ldots \Delta^{(m_k)}_{\mathbf{e_{i_k}},
\ldots, \mathbf{e_{i_k}}} f$ with $t = x_{i_1}^{m_1} \cdots
x_{i_k}^{m_k}$ as in the previous section.

Factoring out $t$, we can write $f$ as
 \[f(x_1, \ldots, x_n) =
 t f_{S(t)}(x_1, \ldots, x_n)+ r(x_1, \ldots, x_n)\]
 where $f_{S(t)}$ and $r$ are uniquely determined such as none of
 the terms in $r$ is divisible by $t$.

We can already give a bound on the degree of $f_t$:
\begin{proposition}
With the notations above, we have $\deg(f_t) \le \deg(f_{S(t)})$. In
particular, if $\deg(t) = \deg(f)-1$ then $f_t$ is linear or
constant.
\end{proposition}
\begin{proof}
When $f$ is a monomial divisible by $t$, we have $\deg(f_t) =
\deg(f_{S(t)})$ by
Theorem~\ref{thm:several-var-fundamental-result-one-monomial}. If
$f$ is a monomial not divisible by $t$, then $f_t=0$.  Finally, for
a general $f$ we use the linearity of the $\Delta$ operator, the
fact that $(f+g)_{S(t)} = f_{S(t)} + g_{S(t)}  $ and the fact that
$\deg(f+g) \le \deg(f) + \deg(g)$, for any polynomials $f$ and $g$.
\end{proof}
The result above is already sufficient for a cube attack. However,
we will give a more refined result shortly, in order to give an
analogue of the main theorem of the classical cube attack (see
Theorems~\ref{thm:cube-classical},
\ref{thm:cube-classical-reformulated} and
\ref{thm:fundamental-p-once}). At first sight we might expect
Theorem~\ref{thm:fundamental-p-once}  to hold here too, namely we
might expect that $f_t(x_1, \ldots, x_n)$ evaluated at $x_{i_j}=0$
for $j = 1, \ldots, k$ equals $f_{S(t)}(x_1, \ldots, x_n)$ evaluated
at $x_{i_j}=1$ for $j = 1, \ldots, k$. However this is not true in
general, as the following counterexample shows:
\begin{example}\label{ex:not-direct-generalisation}
We continue Example~\ref{ex:diff} for $f(x_1, x_2, x_3, x_4) = x_1^5
x_2 + x_1^4 x_3 x_4 + x_4^6$.

We computed $f_{x_1}(x_1, x_2, x_3, x_4)$ in Example~\ref{ex:diff}.
Evaluating at $x_1=0$ we obtain $f_{x_1}(0, x_2, x_3, x_4) =
x_2+x_3x_4$.
On the other hand $f_{S(x_1)}(x_1, x_2, x_3, x_4) =
x_1^4 x_2 + x_1^3 x_3 x_4$. Evaluating at $x_1=1$ we obtain
$f_{S(x_1)}(1, x_2, x_3, x_4) = x_2 +  x_3 x_4$.
Hence we verified that $f_{x_1}(0, x_2, x_3, x_4) = f_{S(x_1)}(1,
x_2, x_3, x_4)$ as expected by Theorem~\ref{thm:fundamental-p-once}.

Differentiating again w.r.t. $x_1$ and evaluating $ f_{x_1^2}(x_1,
x_2, x_3, x_4) $ at $x_1=0$ gives $f_{x_1^2}(0, x_2, x_3, x_4)  = 30
x_2 + 14
 x_3x_4$.
On the other hand $f_{S(x_1^2)}(x_1, x_2, x_3, x_4) = x_1^3 x_2 +
x_1^2 x_3 x_4$, which evaluated at $x_1=1$ gives $ f_{S(x_1^2)}(1,
x_2, x_3, x_4)  =  x_2 + x_3x_4 $

Hence $f_{x_1^2}(0, x_2, x_3, x_4) \neq f_{S(x_1^2)}(1, x_2, x_3,
x_4)$, so Theorem~\ref{thm:fundamental-p-once} cannot be extended in
its current form to the case when we differentiate more than once
w.r.t. one variable. However, note that the two quantities computed
here do contain the same monomials.

Finally, if we differentiate 5 times with respect to $x_1$ we
obtain: $f_{x_1^5}(x_1, x_2, x_3, x_4) = 27 x_2$, whereas
$f_{S(x_1^5)}(x_1, x_2, x_3, x_4) = x_2 $ so again the two
polynomials do not coincide; however  they only differ by
multiplication by a constant.
\end{example}

The correct generalisation of the main theorem of the classical cube
attack is the following:

\begin{theorem}\label{thm:several-var-fundamental-result}
Let $f:\GF(p)^n \rightarrow \GF(p)$ be a polynomial function and $t
= \prod_{j=1}^k x_{i_j}^{m_j}$. Denote
\[
f_t(x_1, \ldots, x_n) = \Delta^{(m_1)}_{\mathbf{e_{i_1}}, \ldots,
\mathbf{e_{i_1}}} \ldots \Delta^{(m_k)}_{\mathbf{e_{i_k}}, \ldots,
\mathbf{e_{i_k}}} f(x_1, \ldots, x_n).
\]
Write $f$ as
\[
f(\mathbf{x}) = t f_{S(t)}(\mathbf{x}) + r(\mathbf{x})
\]
so that none of the monomials in $r$ are divisible by $t$.

Denote by $ \mathbf{v}$ the $n$-tuple having values of 1 in the
positions $i_1, \ldots, i_k$ and indeterminates elsewhere, and by $
\mathbf{u}$ the $n$-tuple having values of 0 in the positions $i_1,
\ldots, i_k$ and indeterminates elsewhere.

 Write
$f_{S(t)} = t_1 g_1 + \ldots + t_ug_u$, where $g_i$ are polynomials
that do not depend on any of the variables $x_{i_1}, \ldots,
x_{i_k}$ and $t_1, \ldots, t_u$ are all the distinct terms in the
variables $x_{i_1}, \ldots, x_{i_k}$ that appear in $f_{S(t)}$.

Then there are constants $c_1, \ldots, c_{u} \in \GF(p)$ such that
$f_t(\mathbf{u})$ equals $c_1g_1 + \ldots + c_ug_u$ (which can also
be viewed as $c_1t_1g_1 + \ldots + c_ut_ug_u$ evaluated at
$\mathbf{v}$). The exact values for the constants $c_i$ can be
determined as follows: if $t_i= x_{i_1}^{\ell_1}\ldots
x_{i_k}^{\ell_k}$, then $c_i = D(m_1+\ell_1, m_1+\ell_1, m_1 )
\ldots D(m_k+\ell_k, m_k+\ell_k, m_k )$ where $D()$ is as defined in
Theorem~\ref{thm:several-var-fundamental-result-one-monomial}.

In particular if $f_{S(t)}$ does not depend on any of the variables
$x_{i_1}, \ldots, x_{i_k}$ then
\[f_t(\mathbf{x}) = m_1! \ldots m_k! f_{S(t)}.\]
\end{theorem}
\begin{proof}
We first use
Theorem~\ref{thm:several-var-fundamental-result-one-monomial} for
individual monomials and then the linearity of the $\Delta$
operator.
 \end{proof}

Again, for the cube attack we are interested in the cases where
$f_t(\mathbf{u})$
is linear:

\begin{corollary}
With the notations of
Theorem~\ref{thm:several-var-fundamental-result},
\[ \deg(f_t(\mathbf{u})) \le \deg(f_{S(t)}(\mathbf{v})).\]
The latter is also equal to the total degree of
$f_{S(t)}(\mathbf{x})$ in the variables $\{x_1, \ldots, x_n\}
\setminus \{ x_{i_1}, \ldots, x_{i_k}  \}$. If
$\deg(f_{S(t)}(\mathbf{v}))=1$ then $f_t(\mathbf{u})$ is linear or
constant.
\end{corollary}

\subsection{Proposed Algorithm for the cube attack in
$\GF(p)$}\label{sec:algorithm}

In this section we give more details of the algorithm, drawing on
the results from previous sections. The main idea of our proposed
attack is that when the degree in one variable is higher than one,
we can differentiate w.r.t. that variable repeatedly, unlike the
cube attacks described in the Section~\ref{sec:cube-classical},
which use differentation at most once for each variable.

We are given a cryptographic ``black box'' function $f(v_1, \ldots
v_m, x_1, \ldots, x_n)$ with $v_i$ being public variables and $x_i$
being secret variables.

{\bf Preprocessing phase}

\begin{enumerate}
  \item Choose a term in the public variables, $t=v_{i_1}^{m_1}\cdots
v_{i_k}^{m_k}$, with $1\le m_i \le p-1$.

  \item Using formula~(\ref{eq:evaluation-diff-p}) in
Proposition~\ref{prop:black-box-eval-p} we evaluate $f_t(\mathbf{0},
\mathbf{x})$ for several choices of the secret variables
$\mathbf{x}$, in order to decide whether, with reasonably high
probability, the total degree of $f_t(\mathbf{0}, \mathbf{x})$ in
$\mathbf{x}$ equals one. (For this one can use the textbook
definition of linearity; namely, for various values of $a,b \in
\GF(p)$ and $\mathbf{y}, \mathbf{z} \in \GF(p)^n$ test whether $a
(f_t(\mathbf{0}, \mathbf{y})-f_t(\mathbf{0}, \mathbf{0})) + b
(f_t(\mathbf{0}, \mathbf{z})-f_t(\mathbf{0}, \mathbf{0})) =
f_t(\mathbf{0}, a\mathbf{y} + b\mathbf{z})-f_t(\mathbf{0},
\mathbf{0})$; in $\GF(p)$ with $p$ large, we will need in general
much fewer linearity tests than in the binary case,
see~\cite{KauRon04}; one can at the same time check whether
$f_t(\mathbf{0}, \mathbf{x})$ is non-constant).

  \item If the decision above is "yes", we determine $f_t(\mathbf{0},
\mathbf{x})$ explicitly, as $f_t(\mathbf{0}, \mathbf{x}) = c_0+
\sum_{i=1}^n c_ix_i$ where $c_0 = f_t(\mathbf{0}, \mathbf{0})$ and
$c_i = f_t(\mathbf{0}, \mathbf{e}_i)-c_0$; we store $(t, c_0, c_1,
\ldots, c_n)$.

  \item Repeat the steps above for different values of $t$ until one obtains
$n$ linearly independent stored tuples $(c_1, \ldots, c_n)$, or
until one runs out of time/memory.
\end{enumerate}

For the heuristic of choosing $t$ one could take into account the
computational cost for a term $t$, see
Proposition~\ref{prop:black-box-eval-p} and the comment following
it. However a full heuristic is beyond the scope of this paper. A
number of optimisations have been proposed for the binary cube
attack; many of them can be transferred to the modulo $p$ case, but
again, this is beyond the scope of this paper.

{\bf Online phase}
\begin{enumerate}
  \item For each $(t, c_0, c_1, \ldots, c_n)$
stored in the preprocessing phase, compute
$f_t(\mathbf{0},\mathbf{x})$ (with $\mathbf{x}$ being now unknown)
 using formula~(\ref{eq:evaluation-diff-p}) in
Proposition~\ref{prop:black-box-eval-p}. Form the linear equation: $
c_1 x_1+ \ldots, +c_n x_n+c_0 =f_t(\mathbf{0},\mathbf{x})$.

  \item  Solve
the system of linear equations thus obtained, determining the secret
variables $x_1, \ldots, x_n$. If the preprocessing phase only
produced $s<n$ equations, then not all the secret variables can de
determined, we would need to do an exhaustive search for $n-s$ of
them.
\end{enumerate}

\begin{remark}
Let $\ell$ be the length of the binary representation of $p$. We can
view each bit of an element in $\GF(p)$ as one binary variable. If
$f$ is a function of $n$ variables over $\GF(p)$, we can also view
it as $\ell$ binary functions in $\ell n$ binary variables. We could
therefore apply the classical (binary) cube attack on these
functions. A rough estimate suggests that the running time for
corresponding cubes will be approximately the same. (Differentiating
$p-1$ times with respect to one variable $x_i$ in $\GF(p)$ takes $p$
evaluations of $f$; differentiating once w.r.t. each of the binary
variables that are components of $x_i$ will take $2^{\ell}$
evaluations of $f$; we have $p \approx 2^{\ell}$.) The chances of
success on a particular cube bear no easy relationship between the
two approaches, because the degree of $f$ and the degrees of the
$\ell$ binary functions are not related in a simple way.

Hence we would argue that in general one cannot tell which of the
attacks will work better, so one should try both. If the cipher has
a structure that would suggest that the degree as polynomial over
$\GF(p)$ is relatively low, then a cube attack over $\GF(p)$ should
certainly be an approach to consider.
\end{remark}

\section{Generalisations to $\GF(p^m)$} \label{sec:proposed-cube-larger-field}

In this section we take our generalisation further, to arbitrary
finite fields $\GF(p^m)$. An important particular case would be
$\GF(2^m)$, as many cryptographic algorithms include operations over
a field of this type.

\subsection{Preliminaries}
We need some known results regarding the values of binomial
coefficients and multinomial coefficients in fields of finite
characteristic.

\begin{theorem}(Kummer's Theorem, \cite[p. 115]{Kum52}) Let $n\ge k \ge 0$
be integers and $p$ a prime. Let $j$ be the highest exponent for
which $\binom{n}{k}$ is divisible by $p^j$. Then $j$ equals the sum
of carries when adding $k$ and $n-k$ as numbers written in base $p$.
\end{theorem}

Kummer's theorem has been generalised to multinomials by various
authors (see for example~\cite{DodPee91} and citations therein).

\begin{theorem}
Let $d, k_1, k_2, \ldots, k_s$ be integers such that $\sum_{i=1}^s
k_i = d$ and $k_i\ge 0$ and let $p$ be a prime. Let $j$ be the
highest exponent for which $\binom{d}{k_1, k_2, \ldots, k_s}$ is
divisible by $p^j$. Then $j$ equals the sum of all the carries when
adding all of $k_1, k_2, \ldots, k_s$ as numbers written in base
$p$.
\end{theorem}

We will be interested in the situations where the 
multinomial coefficients are not zero modulo $p$.

\begin{corollary}\label{cor:multinomial-not-zero}
Let $p$ be a prime.
 The following are equivalent:\\
(i) The multinomial coefficient $\binom{d}{k_1, k_2, \ldots, k_s}$
is not zero modulo
$p$.\\
(ii) There are no carries when adding $k_1, k_2, \ldots, k_s$ as
numbers
written in base $p$. \\
(iii) In base $p$, each digit of $n$ equals to the sum of the digits
of $k_1, k_2, \ldots, k_s$ in the corresponding position.
\end{corollary}

\subsection{Differentiation in $\GF(p^m)$}\label{sec:diff-pm}
When moving from $\GF(p)$ to $\GF(p^m)$ several things work
differently. For a start, differentiating once w.r.t. a variable $x$
may decrease the degree in $x$ by more than one, regardless of the
difference step. For example let us differentiate $x^d$ once. In
$(x+h)^d - x^d$ the coefficient of $x^{d-1}$ is $dh$, so when $d$ is
a multiple of $p$ the degree is strictly less than $d-1$. To examine
the general case we will use Theorem~\ref{thm:diff-formula}, so we
introduce for convenience the following notation:
\[
C_p(d,j, k) = \{ (i_1, \ldots, i_k) | 1\le i_j \le d,
i_1+\ldots+i_k=j, \mbox{ and } \binom{d}{i_1, \ldots, i_{k}, n-j}
\not \equiv 0 \bmod p \}.
\]
Note that Corollary~\ref{cor:multinomial-not-zero} gives a useful
characterisation of this set. We have:
\begin{theorem} \label{thm:diff-degree}
Let $f:\GF(p^m)\rightarrow \GF(p^m)$, $f(x)=x^d$, $d<p^m$. Let
$0<k\le m(p-1)$ and let $h_1,\ldots, h_k\in \GF(p^m)\setminus
\{0\}$. The degree of $\Delta^{(k)}_{h_1\mathbf{e_1}, \ldots,
h_k\mathbf{e_1}} x^d$  is less than or equal to the integer $d'$
computed as follows: write $d$ in base $p$ as $d=d_ud_{u-1}\ldots
d_1d_0$; let $i$ be the highest integer for which $d_0 + d_1+ \cdots
+ d_i \le k$; define $d'_{i+1} = d_{i+1} - (k- (d_0 + d_1+ \cdots +
d_i ))$; finally define $d'$ as the number written in base $p$ as
$d'=d_ud_{u-1}\ldots d_{i+2}d'_{i+1}0\ldots0$ (with $i+1$ zeroes at
the end).

In particular, for $p=2$, the binary representation of the degree
$d'$ is obtained from the binary representation of $d$ by replacing
$k$ of its ones by zeroes, starting from the least significant
digit.
\end{theorem}
\begin{proof}
By Theorem~\ref{thm:diff-formula} the degree $d'$ will be less than
or equal to $d-j$ where $j$ is minimal such that $C_p(d,j,k)\neq
\emptyset$. Using Corollary~\ref{cor:multinomial-not-zero}(iii) we
see that the minimum value for $j$ for given $d$ and $k$ is achieved
by choosing $i_1,\ldots, i_k\ge 1$ as small as possible while
maintaining $\binom{d}{i_1, \ldots i_k, d-j}$ not equal to zero
modulo $p$. This is achieved by choosing $i_1, \ldots i_k$ as
follows: $d_0$ of them will be equal to 1, $d_1$ will be equal to
$p$ (i.e. $10$ in base $p$), $d_2$ will be equal to $p^2$ (i.e.
$100$ in base $p$), $\ldots$, $d_i$ of them will be equal to $p^i$
and finally $k- (d_0 + d_1+ \cdots + d_i)$ will be equal to
$p^{i+1}$. It can be verified that $d'= d- (i_1+\ldots +i_k)$ will
then have the form described in
the theorem statement. 
 \end{proof}

Note that the sum of the digits of a number in base $p$ plays an
important role here. For any non-negative integer $a$ we will
introduce the notation $\S_p(a)$ as being the sum of the digits in
the base $p$ representation of $a$. We define the digit-sum degree
of a univariate polynomial $f$ in a variable $x_i$ as being $\max\{
\S_p(j) | c_j \ne 0 \}$ where $f=\sum_{j=0}^d c_jx_i^j$ with $c_j$
polynomials in the remaining variables. The previous theorem
implies:
\begin{corollary}\label{cor:number-of-diff-before-zero}
Let $f$ be a polynomial function and let $s$ be the digit-sum degree
of $f$ in $x_i$. Then differentiating $f$ a total of $s$ times
w.r.t. $x_i$ will always produce a polynomial function which does
not depend on $x_i$ (possibly the identically zero function).
\end{corollary}
\begin{corollary}
Any function $f:\GF(p^m)\rightarrow \GF(p^m)$ can be differentiated
at most $m(p-1)$ times w.r.t. a given variable before the result
becomes the identically zero function.
\end{corollary}

Is the bound in Corollary~\ref{cor:number-of-diff-before-zero}
tight, in the sense that there are functions which are non-zero
after a number of differentiations equal to their digit-sum degree?
In particular, are there functions which are still non-zero after
$m(p-1)$ differentiations? We will show that this indeed the case,
but only if we choose the $h_i$ carefully. First let us illustrate a
choice of the steps $h_i$ which we need to avoid. By
Proposition~\ref{prop:incl-excl-one-var-p}, if we differentiate $p$
times with all steps equal to 1 the result is identically zero
regardless of the original function $f$:
\begin{eqnarray*}
\Delta^{(p)}_{\mathbf{e_{1}}, \ldots, \mathbf{e_{1}}} f (x_1,
\ldots, x_n) &= &  \sum_{i=0}^{p} (-1)^{p-i}\binom{p}{i}f(x_1+i,
x_2, \ldots, x_n)\\& = & f(x_1+p, x_2,
\ldots, x_n) - f(x_1, x_2, \ldots, x_n)\\
&=& 0.
\end{eqnarray*}
because all the coefficients $\displaystyle \binom{p}{i}$ for
$0<i<p$ are divisible by $p$.

Denote by $b_0, \ldots, b_{m-1}$ a basis of $\GF(p^m)$ viewed as a
$m$-dimensional vector space over $\GF(p)$. We choose the sequence
$h_{1}, \ldots, h_{(p-1)m}$ as follows: $p-1$ values of $b_0$,
followed by $p-1$ values of $b_1$ etc.

 As in Section~\ref{sec:proposed-cube}, we pick a set
of variables and their multiplicities, defining the term $t =
x_{i_1}^{m_1} \cdots x_{i_k}^{m_k}$. For a polynomial function $f$
in $n$ variables, we now define:
\[
f_t(x_1, \ldots, x_n) = \Delta^{(m_1)}_{h_1\mathbf{e_{i_1}}, \ldots,
h_{m_1}\mathbf{e_{i_1}}} \ldots \Delta^{(m_k)}_{h_1\mathbf{e_{i_k}},
\ldots, h_{m_k}\mathbf{e_{i_k}}} f(x_1, \ldots, x_n)
\]
where the sequence $h_{1}, \ldots, h_{m(p-1)}$ has been fixed as
above. We will concentrate on differentiating several times w.r.t.
one variable, $x_1$.

For these choices of $h_i$
Proposition~\ref{prop:diff-incl-excl-formula} becomes:
\begin{proposition}\label{prop:incl-excl-pm}
Let $f:\GF(p^m)^n \rightarrow \GF(p^m)$ be a polynomial function and
let $t=x_1^{m_1}$. Write $m_1 = q (p-1) + r$ with $0< r \le p-1$
(note the slight difference from the usual quotient and reminder, as
here the reminder can be $p-1$ but cannot be 0). Then
\[
f_t(\mathbf{x}) = \sum_{a_0, \dots, a_{q}} (-1)^{{m_1} -
\sum_{i=0}^q a_i} \binom{p-1}{a_0} \cdots
\binom{p-1}{a_{q-1}}\binom{r}{a_q} f(x_1 + a_0 b_0 + \ldots + a_q
b_q, x_2, \ldots,x_n)
\]
where the sum is over all tuples $(a_0, \dots, a_{q})\in \{ 0,
\ldots, p-1 \}^{q} \times \{ 0, \ldots, r \}$.

All the coefficients in the sum above are non-zero. If $f$ is a
``black box'' function, one evaluation of $f_t$ needs $p^q(r+1)$
evaluations of $f$.
\end{proposition}
\begin{proof}
Similar to the proof of Proposition~\ref{prop:black-box-eval-p}.
 \end{proof}

We show next that our choice of $h_i$ is indeed a valid choice in
the sense that there are functions which can be differentiated
$m(p-1)$ times w.r.t. the same variable without becoming zero.

\begin{proposition}\label{prop:h_i-choice-good}
For each $t=x_1^{m_1}$ with
$0\le m_1 \le m(p-1)$ there is at least a function $f:\GF(p^m)^n
\rightarrow \GF(p^m)$ such that $f_t(\mathbf{x})$ is not the
identical zero function. Moreover there is at least a polynomial
function $f$ with digit-sum degree in $x_1$ equal to $m_1$ such that
$f_t(\mathbf{x})$ is a non-zero constant function.
\end{proposition}
\begin{proof}
We will construct a polynomial function $f$ in one variable, $x_1$.
Write $m_1 = q (p-1) + r$ with $0< r \le p-1$ as in
Proposition~\ref{prop:incl-excl-pm}. In the formula in
Proposition~\ref{prop:incl-excl-pm} all the terms in the sum have a
non-zero coefficient, and have distinct arguments for $f$. We will
prescribe the values of $f$ at these evaluation points and then
interpolate $f$. Namely we will prescribe $f(a_0 b_0 + \ldots + a_q
b_q)=0$ for all $(a_0, \dots, a_{q})\in \{ 0, \ldots, p-1 \}^{q}
\times \{ 0, \ldots, r \}$ except $a_0 = \ldots = a_q = 0$, where
$f(0)\ne 0$. The polynomial $f$ is interpolated as
\[
f(x_1) = \prod (x_1 - (a_0 b_0 + \ldots + a_q b_q))
\]
where the product is over $(a_0, \dots, a_{q})\in \{ 0, \ldots, p-1
\}^{q} \times \{ 0, \ldots, r \}$ except $a_0 = \ldots = a_q = 0$.
It can be easily seen that $\deg_{x_1}(f) = (r+1)p^q-1$, so the
representation of its degree in base $p$ consists of a digit of $r$
followed by $q$ digits of $p-1$. Hence $\S_p(\deg_{x_1}(f)) =
q(p-1)+r = m_1$.

On the other hand, we know from Theorem~\ref{thm:diff-degree} that
$f_t$ is a constant (possibly zero). However $f_t(0) = f(0) \ne 0$
by Proposition~\ref{prop:incl-excl-pm} and our choice of
interpolation points. Hence $f_t$ is a non-zero constant. 
\end{proof}

Note that there are other possible valid choices of $h_i$, but we
aimed to keep things simple computationally by using this particular
choice.

\begin{example}
Consider $f(x_1) = x_1^5 \in \GF(9)[x_1]$. We define $h_1=h_2=1$ and
$h_3=h_4=\alpha$, where $\alpha$ is a primitive element of $\GF(9)$.
We compute the third order differential
$\Delta^{(3)}_{\mathbf{e_{1}}, \mathbf{e_{1}}, \alpha\mathbf{e_{1}}}
f(x_1)$ using either Proposition~\ref{prop:diff-incl-excl-formula}
or Theorem~\ref{thm:diff-formula} and obtain:
\[
\Delta^{(3)}_{\mathbf{e_{1}}, \mathbf{e_{1}}, \alpha\mathbf{e_{1}}}
f(x_1) = 2 \alpha^3 + \alpha  = 2\alpha(\alpha+1)(\alpha+2)\neq 0.
\]
\end{example}

Generalising the results of this section to differentiation w.r.t.
several variables is not difficult but the notation becomes
cumbersome. Moreover, we will see in the next subsection that such a
generalisation is not very useful for a practical attack, so we
leave it as an exercise to the reader.

\subsection{Reducing differentiation over $\GF(p^m)$ to
differentiation over $\GF(p)$}

Fix a basis $b_0, \ldots, b_{m-1} \in \GF(p^m)$ for $\GF(p^m)$
viewed as an $m$-dimensional vector space over $\GF(p)$. Any element
$a\in \GF(p^m)$ can be uniquely written as $a=
a_0b_0+\ldots+a_{m-1}b_{m-1}$ with $a_i\in \GF(p)$. Denote by
$\varphi:\GF(p^m)\rightarrow \GF(p)^m$ the vector space isomorphism
defined by $\varphi(a) = (a_0, \ldots, a_{m-1})$; this can be
naturally extended to $\varphi:\GF(p^m)^n\rightarrow \GF(p)^{mn}$ ;
denote by $\pi_j: \GF(p^m)\rightarrow \GF(p)$ the $m$ the projection
homomorphisms defined as $\pi_j(a) = a_j$.

 Let $f:\GF(p^m)^n \rightarrow \GF(p^m)$ be a polynomial
function in $n$ variables $x_1, \ldots, x_n$. By writing $x_i =
x_{i0}b_0 + \ldots + x_{i, m-1} b_{m-1}$ the function $f$ can be
alternatively viewed as a function $\overline{f}:\GF(p)^{nm}
\rightarrow f:\GF(p)^{m}$ defined by $\overline{f} = \varphi^{-1}
\circ f \circ \varphi$:
\[
\begin{array}{ccc}
\GF(p^m)^n & \stackrel{f}{\longrightarrow} &\GF(p^m) \\
\varphi\downarrow & & \downarrow\varphi\\
\GF(p)^{mn} & \stackrel{\overline{f}}{\longrightarrow} &\GF(p)^m \\
\end{array}
\]

Alternatively we can view $f$ as $m$ polynomial (projection)
functions $\overline{f}_0, \ldots, \overline{f}_{m-1}: \GF(p)^{nm}
\rightarrow \GF(p)$, defined by $\overline{f}_i= \varphi^{-1} \circ
f\circ\pi_i $ each in $nm$ variables $x_{ij}$ with $i=1, \ldots, n$
and $j=0, \ldots, m-1$.

\begin{proposition}
With the notations above, the total degree of $\overline{f}_j$ in
the variables $x_{i0}, \ldots, x_{i,m-1}$ is at most the digit-sum
degree of $f$ in $x_i$.
\end{proposition}
\begin{proof}
It suffices to prove the statement for $f=x_1^{d}$. We do it by
induction on the number of digits in the representation of $d_1$ in
base $p$. For one digit, i.e. $d=1, 2, \ldots, p-1$ it is trivially
satisfied. For the induction step, write $d=pd'+d_0$, with $0\le
d_0<p$. Let $A=x_{10}b_0 + x_{11} b_1+ \ldots + x_{1,m-1} b_{m-1}$.
We have:
\[
A^{pd'+d_0} = (A^{d'})^p A^{d_0}
\]
In $\GF(p)$ we have $x^p=x$, so computing the power $p$ of any
polynomial function over $\GF(p)$ does not change the degree in each
variable. Hence the degree in $x_{10}, \ldots, x_{1,m-1}$ of $A^d$
equals at most the degree of $A^{d_0}$ plus the degree of $A^{d'}$,
so by the induction hypothesis, it is at most equal to $d_0$ plus
the sum of the digits of $d'$; this equals the sum of the digits of
$d$. 
 \end{proof}
 Note that in the proposition above $d_i\le
p^m-1$, so the sum of the digits of $d_i$ is at most $(p-1)m$. On
the other hand, $\overline{f}_j$ are polynomial functions over
$\GF(p)$, so they have degree at most $p-1$ in each variable.
Therefore their total degree in the variables $x_{i0}, \ldots,
x_{i,m-1}$ is at most $(p-1)m$, so these facts tie in.

In the next theorem, since we are differentiating functions in a
different number of variables, we will use for elementary vectors
the notation $\mathbf{e^{(n)}_{i}}$  instead of $\mathbf{e_{i}}$ to
clarify the length $n$ of the elementary vector. Note that here the
difference steps are elements in a basis of $\GF(p^m)$, but they do
not need to be in the order prescribed in the discussion at the
beginning of Section~\ref{sec:diff-pm}.
\begin{theorem}\label{thm:cube-attack-reduction}
Let $f:\GF(p^m)^n \rightarrow \GF(p^m)$ be a polynomial function and
let $r_0, \ldots, r_{m-1}$ with $0\le r_i \le p-1$.
\[
\overline{\Delta^{(r_0)}_{b_0\mathbf{e^{(n)}_{1}}, \ldots,
b_0\mathbf{e^{(n)}_{1}}}
 \ldots
 \Delta^{(r_{m-1})}_{b_{m-1}\mathbf{e^{(n)}_{1}}, \ldots,
b_{m-1}\mathbf{e^{(n)}_{1}}}
 f } =
\Delta^{(r_0)}_{\mathbf{e^{(mn)}_{1}}, \ldots,
\mathbf{e^{(mn)}_{1}}}
 \ldots
\Delta^{(r_{m-1})}_{\mathbf{e^{(mn)}_{m-1}}, \ldots,
\mathbf{e^{(mn)}_{m-1}}} \overline{f} \] the latter being a
differentiation $r_0$ times w.r.t $x_{10}$, $r_1$ times w.r.t
$x_{11}$, $\ldots$, $r_{m-1}$ times w.r.t. $x_{1,m-1}$.
\end{theorem}
\begin{proof} Induction on $\sum_{i=0}^{p-1} r_i$. For the base
case, we can assume without loss of generality that $r_0=1$ and
$r_i=0$ for $i\ge 1$. We have:
\begin{eqnarray*}
 \lefteqn{\Delta_{b_0\mathbf{e_{1}}} f(x_1, \ldots, x_n)    }\\
& =& f(x_1+b_0, \ldots, x_n )-f(x_1, \ldots, x_n)
\\ & = & f(x_{10}b_0 + \ldots + x_{1, m-1} b_{m-1} + b_0, \ldots, x_n) -f(x_{10}b_0 + \ldots + x_{1, m-1}
b_{m-1}, \ldots, x_n) \\
& = & f((x_{10}+1)b_0 + \ldots + x_{1, m-1} b_{m-1}, \ldots, x_n)
-f(x_{10}b_0 +
\ldots + x_{1, m-1} b_{m-1}, \ldots, x_n)\\
\end{eqnarray*}
Hence
\begin{eqnarray*}
 \overline{\Delta_{b_0\mathbf{e_{1}}} f(x_1, \ldots, x_n)    }\\
 & = & \overline{f}(x_{10}+1, x_{11}, \ldots,
x_{n,m-1}) -\overline{f}(x_{10}, x_{11}, \ldots, x_{n,m-1})\\
& = & \Delta_{\mathbf{e^{(mn)}_{1}}}\overline{f}(x_{10},
x_{11}, \ldots, x_{n,m-1})\\
\end{eqnarray*}
The inductive step is similar, using  the fact that $\varphi$ is an
isomorphism and $\pi_j$ are homomorphisms. 
 \end{proof}

\subsection{Cube attack in $\GF(p^m)$}
The fundamental result of the cube attack can also be generalised to
$\GF(p^m)$. We will formulate it for differentiation w.r.t. one
variable.
\begin{theorem} \label{thm:cube-higher-field-result-evaluation}
Let $f:\GF(p^m)^n \rightarrow \GF(p^m)$ be a polynomial function of
degree $d_1$ in $x_1$. Write
\[
f(x_1, \ldots, x_n) = \sum_{i=0}^{d_1} g_i(x_2, \ldots, x_n) x_1^i
\]
Let $m_1\ge 1$ and define $t=x_1^{m_1}$. Let $j_1, \ldots, j_r$ be
all those integers between 0 and $d_1$ such that  $\S_p(j_w)\le
m_1$. Then
\[
f_t(0, x_2, \ldots, x_n) = \sum_{w=1}^{r} c_{j_w} g_{j_w}(x_2,
\ldots, x_n),
\]
where $c_{j_1}, \ldots, c_{j_r}$ are constants which only depend on
$h_{1}, \ldots, h_{m_1}$ (and do not depend on $ x_1, \ldots, x_n $)
given by:
\begin{equation}\label{eq:coeff}
c_{j_w} =
 \sum_{ (i_1, \ldots, i_{m_1}) \in C_p(j_w, j_w, m_1)}
\binom{j_w}{i_1, \ldots, i_{m_1}, 0} h_{1}^{i_1}\cdots
h_{m_1}^{i_{m_1}}
\end{equation}

\end{theorem}
\begin{proof}
We apply Theorem~\ref{thm:diff-formula} and the linearity of the
$\Delta$ operator. 
 \end{proof}
\begin{corollary}
Let $f:\GF(p^m)^n \rightarrow \GF(p^m)$ be a polynomial function and
$d_1 = \deg_{x_1}(f)$. Write $f$ as:
\[
f(x_1, \ldots, x_n) = \sum_{i=0}^{d_1} g_i(x_2, \ldots, x_n) x_1^i.
\]
Let $m_1$ be the digit-sum degree of $f$ in $x_1$. Let $j_1, \ldots,
j_r$ be those integers between 0 and $d_1$ for which  $S_p(j_w)=m_1$
and $g_{j_w}\ne 0$. Then $m_1$ is the highest number of times that
we can differentiate $f$ w.r.t. $x_1$ before it becomes identically
zero. Moreover, for $t=x_1^{m_1}$ we have
\[
f_t(x_1, x_2, \ldots, x_n) = \sum_{w=1}^{r} c_{j_w} g_{j_w}(x_2,
\ldots, x_n),
\]
 where $c_{j_w}$ are as defined in (\ref{eq:coeff}).
\end{corollary}
Note that in the Corollary above $f_t$ can be evaluated at any
point, not necessarily at $x_1=0$ like in
Theorem~\ref{thm:cube-higher-field-result-evaluation}, because in
this case $f_t$ does not depend on $x_1$.

\begin{remark}
A cube attack for a polynomial function $f$ over $\GF(p^m)$ can be
developed based on
Theorem~\ref{thm:cube-higher-field-result-evaluation} (generalised
to several variables). However, by using
Theorem~\ref{thm:cube-attack-reduction}, such an attack can be
reduced to an attack in $\GF(p)$ on the polynomial functions
$\overline{f}_0, \ldots, \overline{f}_{m-1}$ simultaneously. In the
cube attack we are looking to differentiate $f$ so that the result
is linear in the secret variables. A polynomial function $f$ is
linear iff all the polynomial functions $\overline{f}_0, \ldots,
\overline{f}_{m-1}$ are linear. However, if we mount an attack in
$\GF(p)$ on $\overline{f}_0, \ldots, \overline{f}_{m-1}$
individually (rather than an attack translated from the attack in
$\GF(p^m)$) there are chances that some of the $\overline{f}_0,
\ldots, \overline{f}_{m-1}$ are linear, even if not all of them are
linear. This suggests that for functions $f$ over $\GF(p^m)$ an
attack in $\GF(p)$ on each component independently is more promising
than an attack in $\GF(p^m)$ on the whole function $f$. Therefore,
in Section~\ref{sec:algorithm} we only described the attack over
$\GF(p)$.
\end{remark}

\section{Conclusion}

We examined higher order differentiation over integers modulo a
prime $p$, as well as over general finite fields of $p^m$ elements,
proving a number of results applicable to cryptographic attacks, and
in particular generalising the fundamental theorem on which the cube
attack is based.

Using these results we proposed a generalisation of the cube attack
to functions over the integers modulo $p$; the main difference to
the binary case is that we can differentiate several times with
respect to the same variable. Such an attack would be particularly
suited to ciphers that use operations modulo $p$ in their internal
structure.

We also show that a further generalisation to general finite fields
$\GF(p^m)$ is possible, but not as promising as the generalisation
to $\GF(p)$, due to the fact that differentiation in $\GF(p^m)$ can
be reduced to differentiation in $\GF(p)$.

\bibliographystyle{plain}

\end{document}